\documentclass[journal]{IEEEtran}
\usepackage{amsmath,amsfonts}
\usepackage{amssymb,amsfonts}
\usepackage{algorithmic}
\usepackage{algorithm}
\usepackage{array}
\usepackage[caption=false,font=normalsize,labelfont=sf,textfont=sf]{subfig}
\usepackage{textcomp}
\usepackage{stfloats}
\usepackage{url}
\usepackage{verbatim}
\usepackage{graphicx}
\usepackage{cite}
\usepackage{bm,bbm}
\PassOptionsToPackage{hidelinks}{hyperref}
\usepackage{orcidlink}
\usepackage{amsthm}
\usepackage{booktabs}
\usepackage{graphicx}
\newcommand{\by}{\bm{y}}
\newcommand{\bg}{\bm{g}}
\newcommand{\bz}{\bm{z}}
\newcommand{\bI}{\bm{I}}
\newcommand{\bA}{\bm{A}}

\newcommand{\bC}{\bm{C}}

\newcommand{\bc}{\bm{c}}
\newcommand{\be}{\bm{e}}
\newcommand{\bB}{\bm{B}}

\newtheorem{proposition}{Proposition}

\newtheorem{remark}{Remark}

\begin{document}

\title{Joint Low-Dimensional Modeling and Sampling Design for Sparse On-Orbit Antenna Pattern Reconstruction}

\author{Yannan Chen,~\IEEEmembership{Memmber,~IEEE}, Qiuchen Liu,~\IEEEmembership{Memmber,~IEEE},Guitong Chen, \\Zezhou Luo, and Lei Huang,~\IEEEmembership{Senior Member,~IEEE}
\thanks{This work was supported by the National Key R\&D Program of China under Grant 2023YFB3905500. A preliminary conference version of part of this work has been accepted for presentation at the 2026 IEEE International Geoscience and Remote Sensing Symposium (IGARSS). (\emph{Corresponding author: Lei Huang.})\\
Yannan Chen, Qiuchen Liu, Guitong Chen, Zezhou Luo and Lei Huang are with the Key Laboratory of Radio Frequency Heterogeneous Integration, Shenzhen University. (E-mail: {chenyannan, lhuang}@szu.edu.cn)\\
}}

\maketitle
\begin{abstract}
    Accurately reconstructing satellite transmit-antenna patterns on orbit is difficult because only sparse directional measurements are available during normal mission operations. This paper develops a cooperative on-orbit pattern-reconstruction framework that converts received calibration power into normalized directional samples and represents the antenna power pattern using a truncated discrete cosine transform (DCT) basis. The resulting low-dimensional model transforms high-dimensional pattern recovery into coefficient estimation, for which a closed-form maximum-likelihood estimator and error characterization are derived. The analysis shows how DCT truncation error, measurement noise, and sampled-basis conditioning jointly affect reconstruction accuracy. For regularly accessible angular sectors, midpoint-uniform sampling provides an information-balanced baseline for the retained DCT modes. For constrained feasible opportunities, D-optimal sampling is used to select informative measurement directions. Simulations verify the accuracy of the angular discretization, the sample efficiency of the truncated-DCT model, and the reconstruction gain of D-optimal sampling under irregular orbit-generated opportunities.
\end{abstract}

\begin{IEEEkeywords}
	Antenna pattern reconstruction, calibration satellite, low-dimensional modeling, D-optimal design, measurement scheduling
\end{IEEEkeywords}

\section{Introduction}

Accurate knowledge of the antenna radiation pattern is essential for ensuring the radiometric stability and geolocation performance of satellite-based microwave sensors. Pre-launch chamber measurements provide an initial characterization, but the on-orbit pattern often deviates from ground-tested results due to on-orbit environmental variations, hardware drift, and antenna-model uncertainties \cite{SAR1992,brautigam2009performance,bachmann2010terrasarx}. These deviations accumulate over the mission lifetime and can significantly degrade retrieval accuracy if left uncorrected. Consequently, reliable on-orbit characterization of the directional antenna pattern has become an important requirement for next-generation spaceborne instruments.

Several operational calibration approaches exist for maintaining the radiometric accuracy of spaceborne microwave sensors. Ground-based calibration sites offer high-quality radiometric references but provide limited angular diversity for antenna-pattern reconstruction \cite{CAVC_2023}. Natural references such as the Moon or cold space serve as stable radiometric targets, yet they do not enable controlled directional sampling required for antenna pattern reconstruction \cite{keihm2025moon}. Cross-calibration between sensors improves radiometric consistency but does not yield instrument-specific pattern information \cite{liu2004new}. Therefore, these approaches are effective for radiometric referencing and consistency monitoring, but they are not sufficient for reconstructing the directional antenna pattern itself.

Conventional antenna pattern measurement techniques, such as compact-range testing \cite{johnson1969compact}, near-field scanning \cite{giordanengo2014fast}, and aerial-platform-based flight measurements \cite{duthoit2017new}, provide accurate antenna characterization before deployment or during dedicated test campaigns. However, they cannot be directly applied after launch, and repeated on-orbit scanning through attitude maneuvers is generally impractical because it consumes fuel and disrupts nominal mission operations. Antenna-pattern-specific calibration has also been investigated in spaceborne synthetic-aperture radar (SAR) missions. Distributed targets have been used to estimate spaceborne SAR antenna patterns \cite{shimada1995technique,guccione2018low}, but the available angular information is determined by the imaging geometry and scene coverage rather than by freely selectable calibration directions. Operational missions such as TerraSAR-X rely on antenna-model calibration, monitoring, and radiometric verification to maintain product accuracy \cite{schwerdt2010final}; however, these procedures are mission-specific and depend on a well-established instrument model, internal calibration information, and dedicated calibration campaigns.

Recent studies have begun to exploit cooperative spaceborne references for on-orbit calibration and antenna-pattern characterization. For spaceborne SAR, calibration-satellite concepts have been investigated to provide external calibration references beyond ground sites \cite{wang2018sarcalnet}, and a nano calibration satellite has been proposed for measuring the elevation antenna pattern of a SAR payload in medium Earth orbit (MEO)\cite{qiu2022inorbit}. More recently, a small measurement satellite flying in a double-cross-helix formation was studied for on-orbit satellite antenna-pattern measurement \cite{mittermayer2024inorbit}. A related line of work treats spaceborne receivers as an antenna range; for example, Cyclone Global Navigation Satellite System (CYGNSS) measurements have been used to characterize Global Positioning System (GPS) transmit power and antenna gain patterns \cite{wang2018gpscygnss,wang2024spaceborne}. These studies demonstrate the feasibility and value of using another spacecraft or spaceborne receiver as a calibration reference. However, they mainly focus on calibration concepts, measurement geometries, or mission-specific data processing. The problem of selecting a very small number of feasible look directions and reconstructing a full transmit pattern from those sparse measurements remains insufficiently addressed, especially when the calibration opportunities are constrained by orbital geometry and mission operations.

In the cooperative scenario considered in this paper, the active microwave microsatellite under test transmits a known calibration signal, and a calibration satellite receives the signal during short visibility windows. After orbit-attitude determination and link-budget compensation, the received power can be converted into directional samples of the microsatellite transmit pattern. However, the accessible look angles are determined jointly by relative orbit geometry, payload attitude, receiver capability, link margin, telemetry schedule, and mission-safety constraints. Therefore, on-orbit transmit-pattern characterization becomes an under-sampled and ill-conditioned inverse problem over a sparse set of physically feasible directional measurements. Sparse antenna-measurement methods have been studied to reduce the number of samples in controlled test facilities by exploiting sparsity in modal or dictionary domains \cite{hofmann2019minimum,bangun2023optimizing}. In contrast, this paper focuses on cooperative on-orbit reconstruction, where the transmit pattern is represented by a fixed low-dimensional DCT subspace and the sampling directions must be selected from physically feasible opportunities constrained by relative orbit geometry and mission operations.

To address these challenges, we propose a unified framework for reconstructing transmit antenna radiation patterns from sparse cooperative on-orbit measurements. The framework converts received calibration power into normalized directional pattern samples and reconstructs the discretized pattern through a low-dimensional DCT coefficient model \cite{ahmed1974discrete} with closed-form maximum-likelihood estimation. This low-dimensional representation reduces the number of required measurements and separates pattern representation from sampling design. For regularly accessible angular sectors, uniform sampling is adopted as an information-balanced baseline, while for constrained feasible opportunities, D-optimal sampling is used to select informative look directions from the available candidate set.

The main contributions are summarized as follows.
\begin{itemize}
    \item We formulate cooperative on-orbit transmit-pattern reconstruction as a sparse inverse problem. The received calibration power is converted into normalized directional pattern samples, and the dense reconstruction grid is distinguished from the physically feasible measurement-opportunity set.

    \item We develop a truncated-DCT model that reduces high-dimensional pattern recovery to low-dimensional coefficient estimation. A closed-form maximum-likelihood estimator and its error characterization are derived to reveal the effects of DCT truncation, measurement noise, and sampled-basis conditioning.

    \item We study sampling design for  feasible opportunities. Uniform sampling is identified as an information-balanced baseline under regular angular accessibility, while D-optimal sampling with a greedy submodular algorithm is used for irregular or incomplete candidate sets.
\end{itemize}

The rest of the paper is organized as follows. Section~\ref{sec:system_model} formulates the cooperative on-orbit pattern reconstruction problem. Section~\ref{sec:low_dimensional_model} presents the low-dimensional pattern model and coefficient estimator. Section~\ref{sec:sampling_design} develops the sampling design strategy. Section~\ref{sec:performance_complexity} analyzes the performance guarantee and computational complexity. Section~\ref{sec:simulation} provides numerical results. Finally, Section~\ref{sec:conclusion} concludes the paper. 

Here and throughout, bold lower-case letters represent vectors, while bold upper-case letters represent matrices. For a matrix $\bm A$, $\bm A^{\top}$, $\bm A^{-1}$, $\operatorname{Tr}(\bm A)$, and $\det(\bm A)$ denote its transpose, inverse, trace, and determinant, respectively. The identity matrix and the all-zero vector are denoted by $\bm I$ and $\bm 0$, respectively. The set of real numbers and the set of nonnegative real numbers are denoted by $\mathbb R$ and $\mathbb R_{+}$, respectively. For a finite set $\mathcal S$, $|\mathcal S|$ denotes its cardinality. For two symmetric matrices $\bm A$ and $\bm B$, $\bm A\succeq\bm B$ means that $\bm A-\bm B$ is positive semidefinite. Moreover, $\mathbb E[\cdot]$, $\|\cdot\|_2$, $\log(\cdot)$, $\kappa(\cdot)$, $\sigma_{\min}(\cdot)$, and $\sigma_{\max}(\cdot)$ denote expectation, Euclidean or spectral norm, natural logarithm, condition number, minimum singular value, and maximum singular value, respectively.

\section{System Model and Problem Formulation}
\label{sec:system_model}

This section introduces the on-orbit transmit-pattern measurement model, provides a sampling-theoretic motivation for discretizing a continuous radiation pattern, and formulates the sparse pattern reconstruction problem.

\subsection{On-Orbit Transmit Power Pattern Measurement Model} 

\begin{figure}[t]
    \centering
    \includegraphics[width=0.9\linewidth]{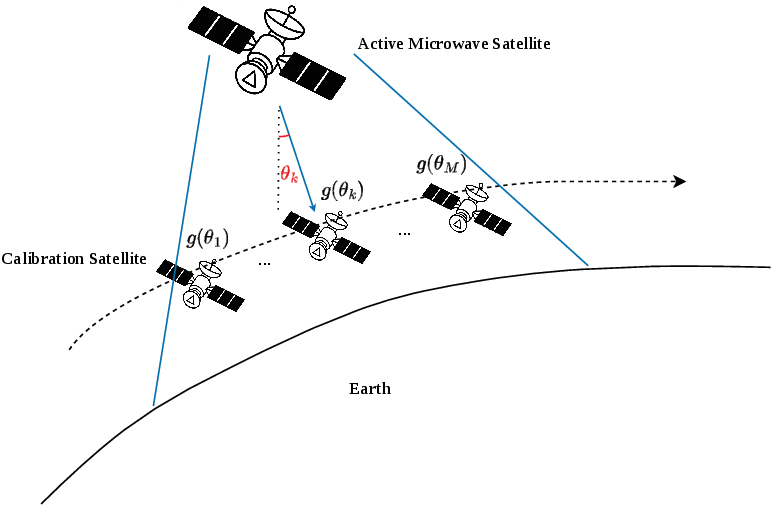}
    \caption{Illustration of the cooperative calibration platform. The platform follows a known trajectory and collects directional measurements.}
    \label{fig:calibration_platform}
\end{figure}

We consider the problem of reconstructing the transmit power radiation pattern of an active microwave satellite over a bounded service angular sector. As illustrated in Fig.~\ref{fig:calibration_platform}, the active microwave satellite transmits calibration signals, while a cooperative calibration satellite follows a known trajectory and receives these signals from different look directions. Each feasible observation provides a received-power measurement, or equivalently a calibrated signal-strength measurement, associated with a transmit direction relative to the payload antenna. 

Let $\Theta$ be the angular region of interest, and the transmit power pattern over this region be a real-valued function $g(\theta)$, where $\theta\in\Theta$. At a feasible transmit direction $\theta$, the received calibration power is modeled as
\begin{align}
    P
    =
    \Gamma g(\theta) + n,
    \label{eq:received_power_model}
\end{align}
where $\Gamma$ collects deterministic link-budget and calibration factors, and $n$ denotes the residual measurement disturbance. In practical spaceborne antenna-range calibration, factors such as propagation loss, receiver gain, geometric corrections, and receiver-chain calibration terms can be estimated or compensated before pattern reconstruction~\cite{wang2024spaceborne}. We therefore use the normalized measurement
\begin{align}
    y
    =
    \frac{P}{\Gamma}
    =
    g(\theta)+z,
    \label{eq:normalized_link_budget_measurement}
\end{align}
where $z=n/\Gamma$ is the equivalent residual error after link-budget normalization.

To represent the unknown continuous pattern, we discretize the service sector into a dense output angular grid
\begin{align}
    \mathcal{C}
    =
    \{\theta_1,\theta_2,\ldots,\theta_M\},
\end{align}
and define the corresponding discretized power pattern as
\begin{align}
    \bg
    =
    [g(\theta_1),g(\theta_2),\ldots,g(\theta_M)]^\top
    \in \mathbb{R}_+^M.
\end{align}
The reconstruction task is therefore cast as estimating $\bg$ from a finite noisy directional measurements.
\subsection{Finite-Aperture Motivation for Angular Discretization}

We now explain why it is reasonable to represent a continuous radiation pattern using a finite angular grid. The proposed reconstruction method does not rely on a specific antenna geometry; nevertheless, a uniform linear array (ULA) provides a useful example for understanding the finite-dimensional structure induced by finite aperture. 

Consider an $N_a$-element ULA with inter-element spacing $d$ and complex excitation coefficients $\{w_n\}_{n=0}^{N_a-1}$. Its complex array factor is \cite{van2002optimum} \begin{align} F(\theta) = \sum_{n=0}^{N_a-1} w_n \exp\!\left( j\frac{2\pi d}{\lambda}n\sin\theta \right). \end{align} Introducing the spatial-frequency variable $u=\sin\theta$, we obtain 
\begin{align} 
    F(u) = \sum_{n=0}^{N_a-1}w_n e^{j\alpha n u}, \qquad \alpha=\frac{2\pi d}{\lambda}. \label{eq:ula_af_u} 
\end{align} 
Thus, the ideal complex array factor is a finite-order exponential polynomial in $u$, whose spatial-frequency complexity is determined by the array aperture. The corresponding power pattern can be written as \begin{align} 
    G(u) = |F(u)|^2 = \sum_{\ell=-(N_a-1)}^{N_a-1} c_\ell e^{j\alpha \ell u}, \label{eq:power_pattern_fourier} 
\end{align} 
where the coefficients $\{c_\ell\}$ are determined by the excitation correlations. 

Eq.~\eqref{eq:power_pattern_fourier} shows that the ideal power pattern is not an arbitrary angular function; its variation in the spatial-frequency domain is limited by the finite aperture. This finite-aperture property supports the use of a sufficiently dense output angular grid to represent the continuous pattern. It also suggests that power patterns over a limited sector are structured and compressible, which motivates a compact representation after discretization. In practical satellite payloads, however, the ideal array structure may be affected by element failures, structural deformation, mutual coupling, and other on-orbit perturbations. We therefore do not impose the Fourier model in \eqref{eq:power_pattern_fourier} as a hard constraint. Instead, it serves only as physical motivation for the low-dimensional representation developed in the next section.

\subsection{Sparse Pattern Reconstruction Problem}

Since the output grid is dense while only a limited number of directional measurements can be collected during a calibration pass, reconstructing the full pattern from the available observations is naturally a sparse reconstruction problem. Based on the normalized measurement model and the output-grid representation above, we now formulate this problem. Let 
\begin{align} 
    \Omega=\{1,2,\ldots,M\} 
\end{align} 
be the index set of the output grid, and $\mathcal S\subseteq\Omega$ be the index set of measured directions with $|\mathcal S|=N_s$. The sampled measurements are collected as 
\begin{align} 
    \by_{\mathcal S} = \bg_{\mathcal S}+\bz, \label{eq:sparse_measurement_model} 
\end{align} 
where $\bg_{\mathcal S}\in\mathbb R^{N_s}$ contains the entries of $\bg$ indexed by $\mathcal S$, and $\bz\sim\mathcal N(\mathbf 0,\sigma^2\bm I)$ is the normalized i.i.d. Gaussian measurement-noise vector.

The goal is to reconstruct the full discretized power pattern $\bg$ from the limited noisy measurements $\by_{\mathcal S}$. Let $\mathcal F(\cdot)$ denote an estimator that maps the sampled measurements $\by_{\mathcal S}$ to an estimate of the full pattern $\bg$. The reconstruction problem can be formulated as \begin{subequations} 
    \label{prob:original} 
    \begin{align} 
        \underset{\mathcal{S},\,\mathcal{F}(\cdot)}{\text{minimize}} &\quad \mathbb E\left[ \left\| \mathcal{F}(\by_{\mathcal S})-\bg \right\|_2^2 \right] \\ \text{subject to} &\quad \mathcal S\subseteq\Omega,\\ &\quad |\mathcal S|=N_s. 
\end{align} \end{subequations} 

Problem~\eqref{prob:original} captures two coupled tasks: selecting informative measurement directions and designing an estimator for the full power pattern. In the following sections, we address this problem by introducing a low-dimensional representation of $\bg$, deriving a coefficient-domain estimator, and then developing sampling strategies for both uniformly accessible angular grids and constrained measurement opportunities.

\section{Low-Dimensional Pattern Representation and Coefficient Estimation}
\label{sec:low_dimensional_model}

The sparse reconstruction problem in \eqref{prob:original} is underdetermined because the number of measurements $N_s$ is much smaller than the output dimension $M$. To make the reconstruction feasible, this section introduces the low-dimensional representation and coefficient-estimation components of the framework. The main idea is to approximate $\bg$ using a small number of basis coefficients, thereby converting the original high-dimensional pattern reconstruction problem into a low-dimensional coefficient estimation problem.

\subsection{DCT-Based Basis Representation}

Motivated by the smoothness and spatial correlation of antenna power patterns over a limited angular sector, we represent the discretized pattern using a reduced basis model
\begin{align}
    \bg
    =
    \bm{\Psi}\bm{\omega}+\be_K,
    \label{eq:low_dim_model}
\end{align}
where $\bm{\Psi}\in\mathbb R^{M\times K}$ is a reduced basis matrix with $K\ll M$, $\bm{\omega}\in\mathbb R^K$ is the coefficient vector, and $\be_K\in\mathbb R^M$ denotes the basis truncation error. When the pattern lies exactly in the subspace spanned by $\bm{\Psi}$, we have $\be_K=\bm 0$.

In this paper, $\bm{\Psi}$ is instantiated using the first $K$ columns of an orthonormal Type-II discrete cosine transform (DCT) matrix~\cite{ahmed1974discrete}. The DCT basis is analytical, orthonormal, and provides strong energy compaction for smooth spatial signals. Let $\bC=[\bc_1,\ldots,\bc_M]\in\mathbb R^{M\times M}$ denote the orthonormal DCT-II matrix, whose entries are given by
\begin{align}
    [\bC]_{m,k}
    =
    \rho_k
    \cos\left[
    \frac{\pi}{M}
    \left(m-\frac{1}{2}\right)(k-1)
    \right],
    \label{eq:dct_basis}
\end{align}
for $m,k=1,\ldots,M$, where
\begin{align}
    \rho_k=
    \begin{cases}
    \sqrt{1/M}, & k=1,\\[1mm]
    \sqrt{2/M}, & k>1.
    \end{cases}
\end{align}
The matrix $\bC$ is orthonormal, i.e., $\bC^\top\bC=\bI$. Applying this transform to $\bg$, we obtain the full DCT coefficient vector
\begin{align}
    \bm{\alpha}
    =
    \bC^\top\bg,
\end{align}
and the exact expansion
\begin{align}
    \bg
    =
    \sum_{k=1}^{M}
    \alpha_k \bc_k.
\end{align}

Although $\bg$ can be exactly represented using all $M$ DCT coefficients, directly estimating all coefficients would still lead to an $M$-dimensional inverse problem. We therefore retain only the first $K$ DCT basis vectors and define
\begin{align}
    \bm{\Psi}
    =
    \bC_{(:,1:K)}
    =
    [\bc_1,\ldots,\bc_K],
    \label{eq:Psi_definition}
\end{align}
with the corresponding coefficient vector
\begin{align}
    \bm{\omega}
    =
    [\alpha_1,\ldots,\alpha_K]^\top.
\end{align}
The resulting $K$-term approximation is
\begin{align}
    \bg_K
    =
    \bm{\Psi}\bm{\omega}
    =
    \sum_{k=1}^{K}\alpha_k\bc_k,
    \label{eq:truncated_dct}
\end{align}
and the truncation error in \eqref{eq:low_dim_model} is
\begin{align}
    \be_K
    =
    \bg-\bg_K
    =
    \sum_{k=K+1}^{M}\alpha_k\bc_k.
    \label{eq:truncation_error_vector}
\end{align}

Since $\bC$ is orthonormal, the relative energy of the discarded components is
\begin{align}
    \epsilon_K
    =
    \frac{\|\be_K\|_2^2}{\|\bg\|_2^2}
    =
    \frac{
        \sum_{k=K+1}^{M}\alpha_k^2
    }{
        \sum_{k=1}^{M}\alpha_k^2
    }.
    \label{eq:truncation_energy}
\end{align}
For a prescribed tolerance $\epsilon\in[0,1]$, an effective dimension can be selected as
\begin{align}
    K_{\min}
    =
    \min\left\{
        K\in\{1,\ldots,M\}
        \;\big|\;
        \epsilon_K\le \epsilon
    \right\}.
    \label{eq:Kmin}
\end{align}
{Fig.~\ref{fig:dct_energy} shows the cumulative DCT coefficient energy for the representative antenna patterns. The energy is strongly concentrated in the low-order coefficients: the first $K=23$ coefficients retain approximately $99.99\%$ of the energy.

\begin{figure}[t]
    \centering
    \includegraphics[width=0.92\linewidth]{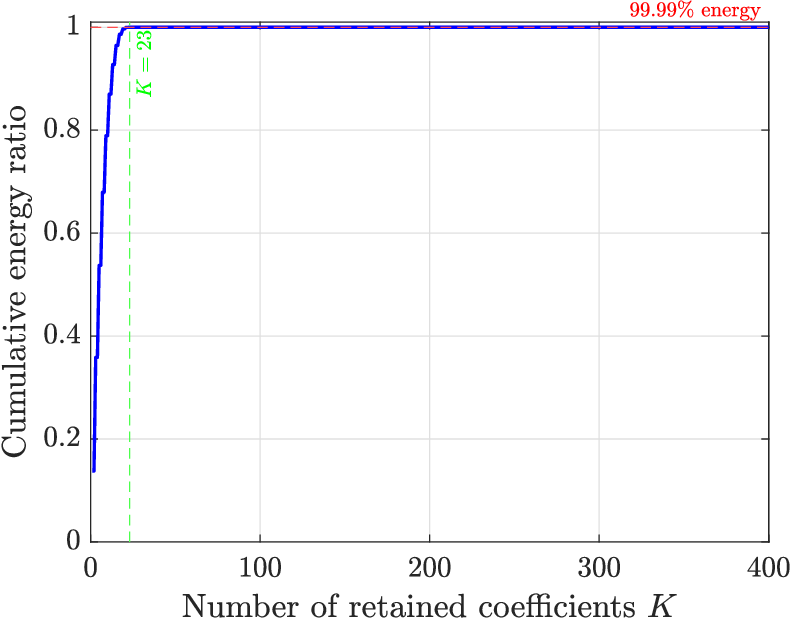}
    \caption{Example of cumulative DCT coefficient energy for the considered perturbed-array pattern model.}
    \label{fig:dct_energy}
\end{figure}}

The reduced representation decreases the number of unknowns from $M$ pattern samples to $K$ coefficients, which is the key reason for its sample efficiency in sparse reconstruction. The resulting sample-efficiency gain over interpolation-based alternatives will be evaluated in Section~\ref{sec:simulation}. In practice, the true on-orbit pattern is unknown before measurement; therefore, the decay of DCT coefficients can be assessed using historical, simulated, or pre-launch reference patterns.

Although the DCT basis is used in this paper for its analytical and orthonormal structure, the subsequent coefficient-domain formulation only requires a prescribed low-dimensional basis. Other bases, such as PCA bases, spherical harmonics, polynomial bases, or data-driven dictionaries, can also be incorporated in the same form.

\subsection{Coefficient-Domain Observation Model and Estimation}

We now derive the observation model in the coefficient domain and estimate the corresponding coefficient vector. Let $\mathcal S\subseteq\Omega$ be the selected sampling set with $|\mathcal S|=N_s$. Substituting \eqref{eq:low_dim_model} into the sampled measurement model gives
\begin{align}
    \by_{\mathcal S}
    =
    \bg_{\mathcal S}
    +
    \bz =
    \bm{\Psi}_{\mathcal S}\bm{\omega}
    +
    (\be_K)_{\mathcal S}
    +
    \bz,
    \label{eq:obs_with_model_error}
\end{align}
where $\bm{\Psi}_{\mathcal S}\in\mathbb R^{N_s\times K}$ is the submatrix obtained by selecting the rows of $\bm{\Psi}$ indexed by $\mathcal S$, and $(\be_K)_{\mathcal S}$ denotes the corresponding entries of the truncation error.

Eq.~\eqref{eq:obs_with_model_error} shows that the observation error consists of two components: the measurement noise $\bz$ and the representation error $(\be_K)_{\mathcal S}$. By collecting these two terms into an effective residual, we obtain the coefficient-domain model
\begin{align}
    \by_{\mathcal S}
    =
    \bm{\Psi}_{\mathcal S}\bm{\omega}
    +
    \tilde{\bz},
    \label{eq:coef_obs_model}
\end{align}
where $\tilde{\bz}=(\be_K)_{\mathcal S}+\bz.$
For coefficient estimation, we use the retained $K$-dimensional component as the signal model and assume that the normalized measurement noise satisfies $\bz\sim\mathcal N(\bm 0,\sigma^2\bI).$
The sampled truncation error $(\be_K)_{\mathcal S}$ is therefore treated as a model-mismatch term rather than as Gaussian noise. When the retained basis dimension $K$ is sufficiently large, this mismatch becomes small, and the least-squares estimator derived below provides a tractable estimate of the dominant coefficient vector $\bm{\omega}$.

Under this Gaussian measurement-noise model, the likelihood associated with the retained $K$-dimensional model is
\begin{align}
    p(\by_{\mathcal S}|\bm{\omega})
    =
    \frac{1}{(2\pi\sigma^2)^{N_s/2}}
    \exp\left(
    -
    \frac{
    \left\|
    \by_{\mathcal S}
    -
    \bm{\Psi}_{\mathcal S}\bm{\omega}
    \right\|_2^2
    }{
    2\sigma^2
    }
    \right).
\end{align}
Maximizing this likelihood is therefore equivalent to solving the least-squares problem
\begin{align}
    \widehat{\bm{\omega}}
    =
    \arg\min_{\bm{\omega}}
    \left\|
    \by_{\mathcal S}
    -
    \bm{\Psi}_{\mathcal S}\bm{\omega}
    \right\|_2^2.
    \label{eq:ls_problem}
\end{align}
If $\bm{\Psi}_{\mathcal S}$ has full column rank, the unique maximum-likelihood estimator is
\begin{align}
    \widehat{\bm{\omega}}
    =
    \left(
    \bm{\Psi}_{\mathcal S}^{\top}
    \bm{\Psi}_{\mathcal S}
    \right)^{-1}
    \bm{\Psi}_{\mathcal S}^{\top}
    \by_{\mathcal S}.
    \label{eq:ml_estimator}
\end{align}
The discretized power pattern is then reconstructed as
\begin{align}
    \widehat{\bg}
    =
    \bm{\Psi}\widehat{\bm{\omega}}.
    \label{eq:pattern_reconstruction}
\end{align}

Eqs.~\eqref{eq:ml_estimator} and \eqref{eq:pattern_reconstruction} show how the sampling set $\mathcal S$ affects the final reconstruction. Specifically, $\mathcal S$ determines the sampled basis matrix $\bm{\Psi}_{\mathcal S}$, which in turn determines the coefficient estimate $\widehat{\bm{\omega}}$ and hence the reconstructed pattern $\widehat{\bg}$. At least $K$ informative measurements are required to identify the $K$-dimensional coefficient vector, and the stability of the estimator depends on the conditioning of $\bm{\Psi}_{\mathcal S}^{\top}\bm{\Psi}_{\mathcal S}$. This relationship is characterized in the next subsection and further used for sampling design.

\subsection{Error Characterization}

We next characterize how the sampling set affects the reconstruction accuracy. The following result summarizes both the statistical error under the retained $K$-dimensional model and the perturbation effect caused by measurement noise and basis truncation.

\begin{proposition}
\label{prop:error_characterization}
Suppose that $\bm{\Psi}_{\mathcal S}$ has full column rank. Under the retained $K$-dimensional model $\bg=\bm{\Psi}\bm{\omega}$ and the Gaussian measurement-noise model, the estimator in \eqref{eq:ml_estimator} is unbiased and satisfies
\begin{align}
    \mathrm{Cov}(\widehat{\bm{\omega}}) = \sigma^2\left(\bm{\Psi}_{\mathcal S}^{\top}\bm{\Psi}_{\mathcal S}\right)^{-1}.
    \label{eq:cov_estimator}
\end{align}
Moreover, since $\widehat{\bg}=\bm{\Psi}\widehat{\bm{\omega}}$ and $\bm{\Psi}^{\top}\bm{\Psi}=\bm I$, the corresponding reconstruction MSE is
\begin{align}
    \mathbb E\left[\|\widehat{\bg}-\bg\|_2^2\right] = \sigma^2\mathrm{Tr}\left[\left(\bm{\Psi}_{\mathcal S}^{\top}\bm{\Psi}_{\mathcal S}\right)^{-1}\right].
    \label{eq:mse_trace}
\end{align}
When the coefficient-domain residual is $\tilde{\bz}=(\be_K)_{\mathcal S}+\bz$, the total reconstruction error satisfies
\begin{align}
    \left\|\widehat{\bg}-\bg\right\|_2 \le \frac{\left\|(\be_K)_{\mathcal S}\right\|_2+\left\|\bz\right\|_2}{\sigma_{\min}(\bm{\Psi}_{\mathcal S})}+\left\|\be_K\right\|_2.
    \label{eq:total_error_bound}
\end{align}
\end{proposition}

\begin{IEEEproof}
Under the retained model, substituting $\by_{\mathcal S}=\bm{\Psi}_{\mathcal S}\bm{\omega}+\bz$ into \eqref{eq:ml_estimator} gives
\begin{align}
    \widehat{\bm{\omega}}-\bm{\omega} = \left(\bm{\Psi}_{\mathcal S}^{\top}\bm{\Psi}_{\mathcal S}\right)^{-1}\bm{\Psi}_{\mathcal S}^{\top}\bz.
    \label{eq:coef_error}
\end{align}
Since $\mathbb E[\bz]=\bm 0$ and $\mathrm{Cov}(\bz)=\sigma^2\bm I$, $\widehat{\bm{\omega}}$ is unbiased and \eqref{eq:cov_estimator} follows. In addition, under the retained model, $\widehat{\bg}-\bg=\bm{\Psi}(\widehat{\bm{\omega}}-\bm{\omega})$. Since $\bm{\Psi}^{\top}\bm{\Psi}=\bm I$, we have $\|\widehat{\bg}-\bg\|_2=\|\widehat{\bm{\omega}}-\bm{\omega}\|_2$, which gives \eqref{eq:mse_trace}.

For the perturbed model $\by_{\mathcal S}=\bm{\Psi}_{\mathcal S}\bm{\omega}+\tilde{\bz}$, the estimator in \eqref{eq:ml_estimator} gives
\begin{align}
    \left\|\widehat{\bm{\omega}}-\bm{\omega}\right\|_2 &= \left\|\left(\bm{\Psi}_{\mathcal S}^{\top}\bm{\Psi}_{\mathcal S}\right)^{-1}\bm{\Psi}_{\mathcal S}^{\top}\tilde{\bz}\right\|_2 \notag\\
    &\le \left\|\left(\bm{\Psi}_{\mathcal S}^{\top}\bm{\Psi}_{\mathcal S}\right)^{-1}\bm{\Psi}_{\mathcal S}^{\top}\right\|_2\left\|\tilde{\bz}\right\|_2 \notag\\
    &\overset{(a)}{=} \frac{\left\|\tilde{\bz}\right\|_2}{\sigma_{\min}(\bm{\Psi}_{\mathcal S})} \notag\\
    &\overset{(b)}{\le} \frac{\left\|(\be_K)_{\mathcal S}\right\|_2+\left\|\bz\right\|_2}{\sigma_{\min}(\bm{\Psi}_{\mathcal S})},
    \label{eq:coef_error_bound}
\end{align}
where $(a)$ follows from the full-column-rank condition and the spectral norm of the left pseudoinverse, and $(b)$ follows from $\tilde{\bz}=(\be_K)_{\mathcal S}+\bz$ and the triangle inequality. Finally, since $\bg=\bm{\Psi}\bm{\omega}+\be_K$ and $\widehat{\bg}=\bm{\Psi}\widehat{\bm{\omega}}$, we have
\begin{align}
    \left\|\widehat{\bg}-\bg\right\|_2 &= \left\|\bm{\Psi}\left(\widehat{\bm{\omega}}-\bm{\omega}\right)-\be_K\right\|_2 \notag\\
    &\le \left\|\bm{\Psi}\left(\widehat{\bm{\omega}}-\bm{\omega}\right)\right\|_2+\left\|\be_K\right\|_2 \notag\\
    &\overset{(c)}{=} \left\|\widehat{\bm{\omega}}-\bm{\omega}\right\|_2+\left\|\be_K\right\|_2,
\end{align}
where $(c)$ uses the orthonormality of the columns of $\bm{\Psi}$. Combining this inequality with \eqref{eq:coef_error_bound} yields \eqref{eq:total_error_bound}.
\end{IEEEproof}

Proposition~\ref{prop:error_characterization} shows that the sampling set affects reconstruction through the information matrix
\begin{align}
    \bm J_{\mathcal S} = \bm{\Psi}_{\mathcal S}^{\top}\bm{\Psi}_{\mathcal S}.
    \label{eq:information_matrix}
\end{align}
Eq.~\eqref{eq:mse_trace} shows that small eigenvalues of $\bm J_{\mathcal S}$ lead to large estimation variance, while \eqref{eq:total_error_bound} shows that small singular values of $\bm{\Psi}_{\mathcal S}$ amplify both measurement noise and model mismatch. Therefore, sampling design should not only ensure that $\bm{\Psi}_{\mathcal S}$ has full column rank, but also improve the numerical stability of the induced least-squares problem.

From a numerical linear algebra perspective, this stability can also be described by the condition number
\begin{align}
    \kappa(\bm{\Psi}_{\mathcal S}) = \frac{\sigma_{\max}(\bm{\Psi}_{\mathcal S})}{\sigma_{\min}(\bm{\Psi}_{\mathcal S})}.
\end{align}
A large condition number indicates that the coefficient estimate can be sensitive to small perturbations in the observations. This motivates an information-matrix-based sampling design. In the next section, we adopt a determinant-based criterion to select informative measurement directions.

\section{Sampling Design}
\label{sec:sampling_design}

Problem~\eqref{prob:original} jointly optimizes the sampling set $\mathcal S$ and the reconstruction rule $\mathcal F(\cdot)$. With the reconstruction rule specified by the coefficient-domain maximum-likelihood estimator in \eqref{eq:ml_estimator} and \eqref{eq:pattern_reconstruction}, this section focuses on the design of $\mathcal S$. We show that uniform sampling is an effective information-balanced strategy when the angular sector is regularly accessible, while D-optimal sampling is used to handle irregular, clustered, or incomplete measurement opportunities.

\subsection{Sampling Design Problem Formulation}

With the estimator fixed as \eqref{eq:ml_estimator} and \eqref{eq:pattern_reconstruction}, the sampling design problem reduces to selecting $\mathcal S$ to minimize the induced reconstruction MSE. According to Proposition~\ref{prop:error_characterization}, under the retained $K$-dimensional model and the Gaussian measurement-noise model, the objective function in Problem~\eqref{prob:original} can be expressed as
\begin{align}
    &\mathbb E\left[\|\mathcal{F}(\by_{\mathcal S})-\bg\|_2^2\right] \notag\\
    =&\mathbb E\left[\|\widehat{\bg}(\mathcal S)-\bg\|_2^2\right]
    =
    \sigma^2
    \mathrm{Tr}
    \left[
    \left(
    \bm{\Psi}_{\mathcal S}^{\top}
    \bm{\Psi}_{\mathcal S}
    \right)^{-1}
    \right].
    \label{eq:sampling_mse_objective}
\end{align}
Thus, problem~\eqref{prob:original} reduces to the following sampling design problem:
\begin{subequations}
\label{prob:a_optimal}
    \begin{align}
    \underset{\mathcal S}{\emph{minimize}}
    &\quad
    \mathrm{Tr}
    \left[
    \left(
    \bm{\Psi}_{\mathcal S}^{\top}
    \bm{\Psi}_{\mathcal S}
    \right)^{-1}
    \right] \\
    \emph{subject to}
    &\quad \mathcal S\subseteq\Omega_c,\\
    &\quad |\mathcal S|=N_s,
    \end{align}
\end{subequations}
where $\Omega_c\subseteq\Omega$ denotes the \emph{feasible candidate set} of measurable directions. For a fully accessible regular grid, $\Omega_c=\Omega$; for constrained on-orbit opportunities, $\Omega_c$ contains only the directions that can be measured.

Problem~\eqref{prob:a_optimal} is a cardinality-constrained subset-selection problem. Exhaustive search requires evaluating $\binom{|\Omega_c|}{N_s}$ possible sampling sets, and each evaluation involves the inverse of a $K\times K$ information matrix. Thus, directly solving \eqref{prob:a_optimal} is computationally prohibitive for large candidate sets. In the following subsections, we consider two practical strategies. When the angular sector is regularly accessible, we show that uniform sampling provides an information-balanced construction for the truncated DCT basis. When the feasible candidate set $\Omega_c$ is irregular, clustered, or incomplete, we adopt a D-optimal surrogate that replaces the trace-inverse objective with a determinant-based information criterion and admits an efficient greedy implementation.

\subsection{Uniform Sampling under Regular Accessibility}

When the angular sector is regularly accessible, uniform sampling provides a simple information-balanced baseline. Since the reconstruction is performed on the $M$-point output grid, uniform sampling means selecting $N_s$ approximately equally spaced rows from the truncated DCT basis matrix $\bm{\Psi}$. The following proposition gives an ideal midpoint-aligned case in which the sampled rows preserve the orthogonality of the retained DCT modes.

\begin{proposition}[Midpoint-aligned uniform sampling of the truncated DCT basis]
\label{prop:uniform_dct}
Let $\bm{\Psi}\in\mathbb R^{M\times K}$ be the truncated DCT basis defined in \eqref{eq:Psi_definition}. Suppose that the sampling set $\mathcal S_{\mathrm u}=\{s_1,\ldots,s_{N_s}\}\subseteq\Omega$ is chosen such that the selected grid locations satisfy the midpoint-alignment condition
\begin{align}
    \frac{s_i-\tfrac{1}{2}}{M}
    =
    \frac{i-\tfrac{1}{2}}{N_s},
    \qquad i=1,\ldots,N_s.
    \label{eq:midpoint_alignment}
\end{align}
If $K\le N_s$, then the sampled basis matrix $\bm{\Psi}_{\mathcal S_{\mathrm u}}$ satisfies
\begin{align}
    \bm{\Psi}_{\mathcal S_{\mathrm u}}^{\top}\bm{\Psi}_{\mathcal S_{\mathrm u}} = \frac{N_s}{M}\bI_K.
    \label{eq:uniform_information}
\end{align}
Consequently, all singular values of $\bm{\Psi}_{\mathcal S_{\mathrm u}}$ are equal to $\sqrt{N_s/M}$, and the coefficient covariance under independent Gaussian measurement noise is
\begin{align}
    \operatorname{Cov}(\widehat{\bm{\omega}})=\sigma^2\frac{M}{N_s}\bI_K.
    \label{eq:uniform_covariance}
\end{align}
\end{proposition}

\begin{IEEEproof}
For the selected rows, the $(i,k)$th entry of $\bm{\Psi}_{\mathcal S_{\mathrm u}}$ is
\begin{align}
    [\bm{\Psi}_{\mathcal S_{\mathrm u}}]_{i,k}=\rho_k\cos\left[\pi(k-1)\xi_i\right],
\end{align}
where $\rho_1=1/\sqrt M$ and $\rho_k=\sqrt{2/M}$ for $k>1$. Let $p=k-1$ and $q=\ell-1$. Under the midpoint-aligned sampling locations, the midpoint cosine orthogonality identity gives
\begin{align}
    \sum_{i=1}^{N_s}\cos(\pi p\xi_i)\cos(\pi q\xi_i)
    =
    \begin{cases}
    N_s, & p=q=0,\\
    N_s/2, & p=q>0,\\
    0, & p\ne q.
    \end{cases}
\end{align}
Applying the normalization factors $\rho_k$ yields \eqref{eq:uniform_information}. The singular-value and covariance results follow immediately from \eqref{eq:uniform_information} and \eqref{eq:cov_estimator}.
\end{IEEEproof}

Proposition~\ref{prop:uniform_dct} provides an ideal explanation for why uniform sampling is effective in regularly accessible angular sectors. In practice, the selected directions must lie on the finite $M$-point output grid, and exact midpoint alignment may not always hold. In such cases, selecting the grid points closest to the ideal midpoint locations yields an approximately isotropic information matrix when the output grid is sufficiently dense. Therefore, uniform sampling is adopted as the basic reconstruction scheme for regular coverage. Optimized sampling becomes most useful when this favorable uniform geometry cannot be realized, for example when the feasible opportunities are irregular, clustered, or incomplete.

\subsection{D-Optimal Sampling under Constrained Accessibility}

When the feasible candidate set $\Omega_c$ is irregular, clustered, or incomplete, the uniform construction in Proposition~\ref{prop:uniform_dct} may no longer be feasible. In this case, sampling design becomes an active subset-selection problem over the available candidates. Since the A-optimal objective in \eqref{prob:a_optimal} is difficult to optimize directly and lacks a simple greedy marginal gain, we adopt the D-optimal criterion as a tractable surrogate:

\begin{subequations}
\label{prob:d_optimal}
    \begin{align}
    \underset{\mathcal S}{\text{maximize}}
    &\quad
    f(\mathcal S)
    =
    \log\det
    \left(
    \bm{\Psi}_{\mathcal S}^{\top}
    \bm{\Psi}_{\mathcal S}
    +
    \varepsilon\bI
    \right) \\
    \text{subject to}
    &\quad
    \mathcal S\subseteq\Omega_c,\\
    &\quad
    |\mathcal S|=N_s,
    \end{align}
\end{subequations}
where $\varepsilon>0$ is a small regularization parameter introduced for numerical stability. The regularization also ensures that the objective is well defined during the early stages of greedy selection, when fewer than $K$ directions may have been selected.

The D-optimal criterion maximizes the determinant of the regularized information matrix. Without the regularization term, the determinant satisfies
\begin{align}
    \det
    \left(
    \bm{\Psi}_{\mathcal S}^{\top}
    \bm{\Psi}_{\mathcal S}
    \right)
    =
    \prod_{k=1}^{K}
    \sigma_k^2(\bm{\Psi}_{\mathcal S}),
\end{align}
where $\sigma_k(\bm{\Psi}_{\mathcal S})$ denotes the $k$th singular value of the sampled basis matrix. Thus, D-optimality favors sampling sets that enlarge the singular values of $\bm{\Psi}_{\mathcal S}$ and penalizes nearly rank-deficient selections. Although it is not equivalent to directly minimizing the A-optimal trace-inverse objective, it promotes an informative and numerically stable information matrix and admits an efficient greedy implementation.

\begin{remark}[Statistical interpretation]
In the linear Gaussian model, the Fisher information matrix for $\bm{\omega}$ is proportional to $\bm{\Psi}_{\mathcal S}^{\top}\bm{\Psi}_{\mathcal S}$. Therefore, maximizing the determinant of this matrix minimizes the volume of the confidence ellipsoid of the coefficient estimate, providing a statistical justification for the D-optimal sampling design.
\end{remark}

\begin{remark}[Model-aware rather than pattern-adaptive sampling]
The proposed sampling design depends on the basis matrix $\bm{\Psi}$, but not on the unknown coefficient vector $\bm{\omega}$ or the unknown pattern $\bg$. This is intentional: before calibration measurements are acquired, the actual on-orbit radiation pattern is unavailable. The proposed design is therefore model-aware rather than pattern-adaptive, selecting measurement directions that make the retained coefficient vector identifiable and stable to estimate for a broad class of patterns represented by the DCT subspace.
\end{remark}

The D-optimal sampling problem in \eqref{prob:d_optimal} is combinatorial. To obtain an efficient solution, we employ a greedy strategy that sequentially selects the candidate direction with the largest marginal improvement in the D-optimal objective.

For a current sampling set $\mathcal S$, define the regularized information matrix as
\begin{align}
    \bA_{\mathcal S}
    =
    \bm{\Psi}_{\mathcal S}^{\top}
    \bm{\Psi}_{\mathcal S}
    +
    \varepsilon\bI.
\end{align}
Let $\bar{\bm{\psi}}_j\in\mathbb R^K$ denote the transpose of the $j$th row of $\bm{\Psi}$, i.e., the feature vector associated with candidate direction $\theta_j$. If index $j\notin\mathcal S$ is added, the information matrix becomes
\begin{align}
    \bA_{\mathcal S\cup\{j\}}
    =
    \bA_{\mathcal S}
    +
    \bar{\bm{\psi}}_j
    \bar{\bm{\psi}}_j^{\top}.
\end{align}
The corresponding marginal gain is
\begin{align}
    \Delta_j(\mathcal S)
    &=
    f(\mathcal S\cup\{j\})-f(\mathcal S) \nonumber\\
    &=
    \log\det
    \left(
    \bA_{\mathcal S}
    +
    \bar{\bm{\psi}}_j
    \bar{\bm{\psi}}_j^{\top}
    \right)
    -
    \log\det
    \left(
    \bA_{\mathcal S}
    \right).
\end{align}
Using the matrix determinant lemma, this gain simplifies to
\begin{align}
    \Delta_j(\mathcal S)
    =
    \log
    \left(
    1+
    \bar{\bm{\psi}}_j^{\top}
    \bA_{\mathcal S}^{-1}
    \bar{\bm{\psi}}_j
    \right).
    \label{eq:marginal_gain}
\end{align}
Thus, the next sampling index is selected as
\begin{align}
    j^\star
    =
    \arg\max_{j\in\Omega\setminus\mathcal S}
    \Delta_j(\mathcal S).
    \label{eq:greedy_choice}
\end{align}
After selecting $j^\star$, the sampling set and information matrix are updated as
\begin{align}
    \mathcal S
    &\leftarrow
    \mathcal S\cup\{j^\star\},\\
    \bA_{\mathcal S}
    &\leftarrow
    \bA_{\mathcal S}
    +
    \bar{\bm{\psi}}_{j^\star}
    \bar{\bm{\psi}}_{j^\star}^{\top}.
\end{align}
The procedure repeats until $N_s$ sampling directions are selected. The complete algorithm is summarized in Algorithm~\ref{alg:greedy_d_optimal}.

\begin{algorithm}[t!]
\caption{Greedy D-Optimal Sampling}
\label{alg:greedy_d_optimal}
\begin{algorithmic}[1]
\STATE \textbf{Input:} Basis matrix $\bm{\Psi}$, sampling budget $N_s$, regularization parameter $\varepsilon$.
\STATE Initialize $\mathcal S\gets\varnothing$, $\bA\gets\varepsilon\bI_K$.
\FOR{$t=1$ to $N_s$}
    \FOR{each $j\in\Omega_c\setminus\mathcal S$}
        \STATE Compute
        \[
        \Delta_j=
        \log
        \left(
        1+
        \bar{\bm{\psi}}_j^{\top}
        \bA^{-1}
        \bar{\bm{\psi}}_j
        \right).
        \]
    \ENDFOR
    \STATE Select $j^\star=\arg\max_{j\in\Omega_c\setminus\mathcal S}\Delta_j$.
    \STATE Update $\mathcal S\gets\mathcal S\cup\{j^\star\}$.
    \STATE Update $\bA\gets\bA+\bar{\bm{\psi}}_{j^\star}\bar{\bm{\psi}}_{j^\star}^{\top}$.
\ENDFOR
\STATE \textbf{Output:} Sampling set $\mathcal S$.
\end{algorithmic}
\end{algorithm}

\subsection{Overall Reconstruction Procedure}
\label{subsec:overall_procedure}
\begin{figure}[t]
    \centering
    \includegraphics[width=\linewidth]{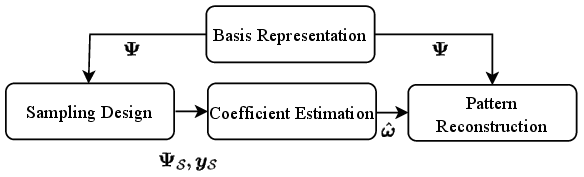}
    \caption{Overview of the proposed on-orbit pattern-reconstruction framework under limited sampling resources.}
    \label{fig:framework}
\end{figure}

Fig.~\ref{fig:framework} summarizes the proposed on-orbit pattern reconstruction framework. The framework first represents the $M$-point antenna power pattern using the first $K$ DCT basis vectors, thereby reducing the number of unknowns from $M$ pattern samples to $K$ basis coefficients. Given a sampling set $\mathcal S$, the selected measurements are then used to estimate the coefficient vector through the coefficient-domain estimator in \eqref{eq:ml_estimator}, and the full-grid pattern is reconstructed by the reduced basis expansion in \eqref{eq:pattern_reconstruction}.

The sampling set $\mathcal S$ is designed according to the accessibility of measurement opportunities. Under regular angular accessibility, uniform sampling provides a well-conditioned and information-balanced choice for the truncated DCT basis. Under constrained or irregular accessibility, D-optimal sampling selects informative feasible directions by maximizing the regularized information volume. Therefore, the overall procedure converts sparse on-orbit pattern reconstruction from direct estimation of $M$ pattern samples into low-dimensional coefficient estimation from carefully selected measurements.

\section{Performance Guarantee and Complexity Analysis}
\label{sec:performance_complexity}
\label{sec:guarantee_complexity}

This section analyzes the structural properties of the proposed D-optimal sampling objective and the computational complexity of the overall reconstruction framework. We first show that the normalized D-optimal objective is monotone and submodular, which leads to a standard greedy approximation guarantee. We then discuss the computational cost of basis construction, greedy sampling, and coefficient estimation.

\subsection{Submodularity and Greedy Approximation Guarantee}

Recall the D-optimal objective
\begin{align}
    f(\mathcal S)
    =
    \log\det
    \left(
    \bm{\Psi}_{\mathcal S}^{\top}
    \bm{\Psi}_{\mathcal S}
    +
    \varepsilon\bI
    \right),
    \label{eq:f_original}
\end{align}
where $\varepsilon>0$ guarantees positive definiteness for any sampling set $\mathcal S$. To apply the standard greedy approximation result, we use the normalized objective
\begin{align}
    F(\mathcal S)
    &=
    f(\mathcal S)-f(\emptyset) \nonumber\\
    &=
    \log\det
    \left(
    \bI
    +
    \varepsilon^{-1}
    \bm{\Psi}_{\mathcal S}^{\top}
    \bm{\Psi}_{\mathcal S}
    \right).
    \label{eq:normalized_objective}
\end{align}
The functions $f$ and $F$ have the same maximizers because they differ only by a constant. Moreover, $F(\emptyset)=0$ and $F(\mathcal S)\ge 0$ for all $\mathcal S\subseteq\Omega_c$.

We next show that $F$ is monotone and submodular. For any $\mathcal S\subseteq\Omega$, define
\begin{align}
    \bB_{\mathcal S}
    =
    \bI
    +
    \varepsilon^{-1}
    \bm{\Psi}_{\mathcal S}^{\top}
    \bm{\Psi}_{\mathcal S}.
\end{align}
If $\mathcal S\subseteq\mathcal T$, then
\begin{align}
    \bm{\Psi}_{\mathcal T}^{\top}\bm{\Psi}_{\mathcal T}
    =
    \bm{\Psi}_{\mathcal S}^{\top}\bm{\Psi}_{\mathcal S}
    +
    \sum_{j\in\mathcal T\setminus\mathcal S}
    \bar{\bm{\psi}}_j
    \bar{\bm{\psi}}_j^{\top}
    \succeq
    \bm{\Psi}_{\mathcal S}^{\top}\bm{\Psi}_{\mathcal S}.
\end{align}
Therefore,
\begin{align}
    \bB_{\mathcal T}\succeq \bB_{\mathcal S}.
\end{align}
Since the log-determinant function is monotone over the positive definite cone, we have
\begin{align}
    F(\mathcal T)\ge F(\mathcal S),
\end{align}
which proves monotonicity.

To establish submodularity, consider $\mathcal S\subseteq\mathcal T\subseteq\Omega_c$ and $j\in\Omega_c\setminus\mathcal T$. The marginal gain of adding $j$ to $\mathcal S$ is
\begin{align}
    F(\mathcal S\cup\{j\})-F(\mathcal S)
    =
    \log
    \left(
    1+
    \varepsilon^{-1}
    \bar{\bm{\psi}}_j^{\top}
    \bB_{\mathcal S}^{-1}
    \bar{\bm{\psi}}_j
    \right),
    \label{eq:normalized_marginal_gain}
\end{align}
where the matrix determinant lemma is used. Since $\mathcal S\subseteq\mathcal T$, we have
\begin{align}
    \bB_{\mathcal T}\succeq \bB_{\mathcal S},
    \qquad
    \bB_{\mathcal T}^{-1}\preceq \bB_{\mathcal S}^{-1}.
\end{align}
Thus,
\begin{align}
    \bar{\bm{\psi}}_j^{\top}
    \bB_{\mathcal T}^{-1}
    \bar{\bm{\psi}}_j
    \le
    \bar{\bm{\psi}}_j^{\top}
    \bB_{\mathcal S}^{-1}
    \bar{\bm{\psi}}_j,
\end{align}
which implies
\begin{align}
    F(\mathcal S\cup\{j\})-F(\mathcal S)
    \ge
    F(\mathcal T\cup\{j\})-F(\mathcal T).
\end{align}
Therefore, $F$ is submodular.

\begin{proposition}[Greedy Approximation Guarantee]
\label{prop:greedy_guarantee}
Let $\mathcal S_{\mathrm{greedy}}$ be the sampling set returned by Algorithm~\ref{alg:greedy_d_optimal}, and let $\mathcal S^\star$ be an optimal solution to the D-optimal sampling problem under the cardinality constraint $|\mathcal S|=N_s$. Then
\begin{align}
    F(\mathcal S_{\mathrm{greedy}})
    \ge
    \left(1-\frac{1}{e}\right)
    F(\mathcal S^\star),
    \label{eq:greedy_bound}
\end{align}
where $F$ is the normalized D-optimal objective defined in \eqref{eq:normalized_objective}.
\end{proposition}

\begin{IEEEproof}
As shown above, $F$ is normalized, nonnegative, monotone, and submodular. Under a cardinality constraint, the classical result for greedy maximization of monotone submodular functions gives the $(1-1/e)$ approximation guarantee. This completes the proof.
\end{IEEEproof}

Since $F$ and $f$ differ only by the constant $f(\emptyset)$, maximizing $F$ is equivalent to maximizing the original D-optimal objective $f$. Proposition~\ref{prop:greedy_guarantee} therefore guarantees that the proposed greedy sampler achieves a constant-factor approximation to the optimal D-optimal sampling design in terms of the normalized information gain.

\subsection{Computational Complexity}

The computational cost of the proposed framework mainly consists of basis construction, sampling-set design, coefficient estimation, and pattern reconstruction. The reduced DCT basis $\bm{\Psi}\in\mathbb R^{M\times K}$ requires $\mathcal O(MK)$ operations to construct and can be computed offline.

For regular coverage, uniform sampling only generates $N_s$ approximately equally spaced indices on the $M$-point output grid, with complexity $\mathcal O(N_s)$. This cost is negligible compared with coefficient estimation and pattern reconstruction.

For constrained candidate sets, D-optimal greedy sampling is more expensive. Let $C=|\Omega_c|$ denote the number of feasible candidate directions. If the inverse of the current regularized information matrix is maintained, each marginal-gain evaluation in \eqref{eq:normalized_marginal_gain} costs $\mathcal O(K^2)$. Therefore, the total complexity of greedy sampling over $N_s$ iterations is $\mathcal O(N_s C K^2)$. In the fully accessible case, $C=M$. After each selected index $j^\star$, the inverse matrix can be updated using the Sherman--Morrison formula:
\begin{align}
    \left(
    \bA+
    \bar{\bm{\psi}}_{j^\star}
    \bar{\bm{\psi}}_{j^\star}^{\top}
    \right)^{-1}
    =
    \bA^{-1}
    -
    \frac{
    \bA^{-1}
    \bar{\bm{\psi}}_{j^\star}
    \bar{\bm{\psi}}_{j^\star}^{\top}
    \bA^{-1}
    }{
    1+
    \bar{\bm{\psi}}_{j^\star}^{\top}
    \bA^{-1}
    \bar{\bm{\psi}}_{j^\star}
    }.
    \label{eq:sherman_morrison}
\end{align}
This rank-one update costs $\mathcal O(K^2)$ per iteration and does not change the dominant complexity.

Once the sampling set is fixed, coefficient estimation using \eqref{eq:ml_estimator} requires $\mathcal O(N_sK^2+K^3)$ operations. If the full $M$-point pattern is reconstructed through \eqref{eq:pattern_reconstruction}, the additional cost is $\mathcal O(MK)$. Since $K\ll M$, the online reconstruction remains low-dimensional.

Overall, uniform sampling has negligible design complexity, whereas D-optimal greedy sampling has design complexity $\mathcal O(N_s|\Omega_c|K^2)$. The basis construction and sampling-set design can be performed offline. During the calibration stage, the system only needs to execute the selected measurement opportunities and solve a low-dimensional coefficient-estimation problem.

\section{Simulation}
\label{sec:simulation}

{
In this section, we evaluate angular-grid discretization, low-dimensional reconstruction, and sampling design. A uniform angular grid is used for the regular-access experiments, while the constrained-opportunity experiment propagates two LEO satellites and applies geometry and link-budget screening to generate the candidate angles.
}

\subsection{Simulation Setup}

We consider a uniform linear array (ULA) with $N_a=64$ elements and an inter-element spacing of $d=0.5\lambda$.
The radiation pattern is evaluated over a limited angular region of $\pm10^\circ$, reflecting the narrow service coverage of satellite payloads, and is discretized into $M=400$ angular grid points.

The array factor is given by
\begin{equation}
\mathrm{AF}(\theta)
=
\sum_{n=0}^{N_a-1}
w_n
\exp\big(
\mathrm{j} k d n \sin\theta
\big),
\end{equation}
where $k=2\pi/\lambda$ and $w_n$ denotes the complex excitation coefficient of the $n$th element.
To emulate realistic on-orbit pattern degradation, $w_n$ is subject to random amplitude and phase perturbations modeled as zero-mean Gaussian variables with standard deviations $\sigma_a=0.1$ and $\sigma_\phi=0.1$, respectively, and $15\%$ of the elements are randomly deactivated. The resulting normalized power pattern $G(\theta)=|\mathrm{AF}(\theta)|^2$ serves as the ground truth for reconstruction. The discretized pattern is represented using a low-dimensional DCT basis with $K=24\ll M$ retained components, and the number of sampled directions is set to $N_s=K$.
Additive Gaussian noise with variance $\sigma^2=10^{-8}$ is applied to the sampled measurements, and a regularization parameter $\varepsilon=10^{-6}$ is used to ensure numerical stability. Two baseline methods, namely \emph{uniform sampling} and \emph{random sampling}, are used for comparison.

{Unless otherwise stated, the reported MSE is computed on the normalized linear power pattern and expressed as $10\log_{10}(\mathrm{MSE})$. All Monte Carlo statistics reported in this section are computed over 1000 trials. In every trial, the amplitude errors, phase errors, failed-element set, and measurement noise are independently regenerated.}

\subsection{Angular-Grid Discretization} 
Before evaluating sparse reconstruction, we first verify that the adopted angular output grid is sufficiently dense. A 6400-point midpoint angular grid is used as the reference. For each candidate grid size $M$, the power pattern is evaluated on the corresponding midpoint grid and interpolated back to the reference angles using PCHIP. The interpolation is used only for this discretization test. The grid-induced representation error is computed as 
\begin{align} 
    E_{\mathrm{grid}}(M) = \frac{1}{N_{\mathrm{ref}}} \sum_{q=1}^{N_{\mathrm{ref}}} \left[ \widetilde G_M(\theta_q) - G_{\mathrm{ref}}(\theta_q) \right]^2, \label{eq:grid_discretization_error} 
\end{align} 
where $G_{\mathrm{ref}}(\theta_q)$ denotes the reference pattern evaluated on the dense grid and $\widetilde G_M(\theta_q)$ denotes the coarse-grid pattern interpolated back to the same reference angle. \begin{table}[t] 
    \centering 
    \caption{Angular-grid representation error relative to a 6400-point reference grid. The 5th and 95th percentiles are computed over Monte Carlo trials.} \label{tab:discretization_error} 
    \begin{tabular}{cccc} 
        \toprule 
        $M$ & Median MSE (dB) & 5th perc. (dB) & 95th perc. (dB)  \\    
        \midrule 
        50 & -46.37 & -47.42 & -45.23 \\ 
        100 & -61.45 & -63.04 & -60.43 \\ 
        200 & -77.13 & -79.85 & -75.50 \\ 
        400 & -93.39 & -96.94 & -90.69 \\ 
        800 & -109.39 & -112.22 & -105.94 \\ 
        \bottomrule 
    \end{tabular} 
\end{table} 

Table~\ref{tab:discretization_error} shows that the representation error decreases rapidly as the angular grid becomes denser. At the adopted value $M=400$, the median discretization MSE is $-93.39$ dB, and the 95th percentile is $-90.69$ dB, indicating that the angular-grid representation error is sufficiently small for the present numerical study. 

\subsection{Low-Dimensional Reconstruction}

We now evaluate the sample efficiency gained by the low-dimensional DCT representation.  All methods use the same uniformly spaced sample locations, noisy measurements, and ground-truth patterns. For the DCT method, the retained dimension is matched to the measurement budget, i.e., $K=N_s$, so that the number of measurements also determines the number of recoverable low-frequency DCT modes.

\begin{figure}[t]
    \centering
    \includegraphics[width=0.98\linewidth]{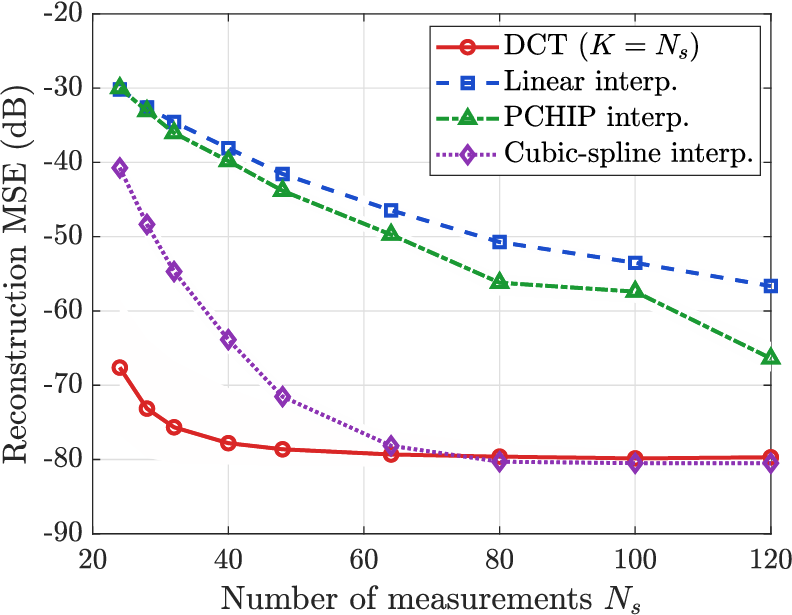}
    \caption{Sample-efficiency comparison between truncated-DCT reconstruction and direct interpolation. For each measurement budget, the DCT method uses $K=N_s$ retained coefficients. All methods use the same uniform sample locations, noisy measurements, and ground-truth patterns.}
    \label{fig:low_dim_benchmark}
\end{figure}

\begin{table}[t]
\centering
\caption{Minimum samples required to reach a target median MSE. A dash indicates that the target is not reached with at most 120 samples.}
\label{tab:sample_efficiency}
\begin{tabular}{ccccc}
\toprule
Target & DCT & Linear & PCHIP & Spline \\
\midrule
$-60$ dB & 24 & -- & 120 & 40 \\
$-70$ dB & 28 & -- & -- & 48 \\
$-75$ dB & 32 & -- & -- & 64 \\
\bottomrule
\end{tabular}
\end{table}

Fig.~\ref{fig:low_dim_benchmark} and Table~\ref{tab:sample_efficiency} show that the truncated-DCT estimator is substantially more sample-efficient than direct interpolation in the sparse-measurement regime. At $N_s=24$, the truncated-DCT estimator attains a median MSE of $-67.62$ dB, compared with $-30.16$, $-29.98$, and $-40.73$ dB for linear, PCHIP, and cubic-spline interpolation, respectively. The DCT model reaches the $-70$-dB target with 28 samples and the $-75$-dB target with 32 samples, whereas cubic spline requires 48 and 64 samples for the same targets. Linear interpolation and PCHIP do not reach $-70$ dB within 120 samples. At dense sampling, cubic spline approaches the same $-80$-dB error range, indicating that the DCT advantage mainly comes from sparse-sample efficiency rather than a universal advantage over interpolation under abundant measurements.

\subsection{Uniform Sampling under Regular Accessibility}

We next examine the role of sampling design when the complete regular angular grid is accessible. This experiment compares uniform sampling, D-optimal sampling, and random sampling with $N_s=K=24$ over the same Monte Carlo trials. In addition to reconstruction MSE, we record the condition number $\kappa(\bm{\Psi}_{\mathcal S})$ and the regularized D-optimal metric $\log\det(\bm{\Psi}_{\mathcal S}^{\top}\bm{\Psi}_{\mathcal S}+\varepsilon\bm I)$. The goal is to verify whether uniform sampling already provides a well-conditioned sampled basis matrix under regular accessibility, as suggested by Proposition~\ref{prop:uniform_dct}.

\begin{figure}[t]
    \centering
    \includegraphics[width=\linewidth]{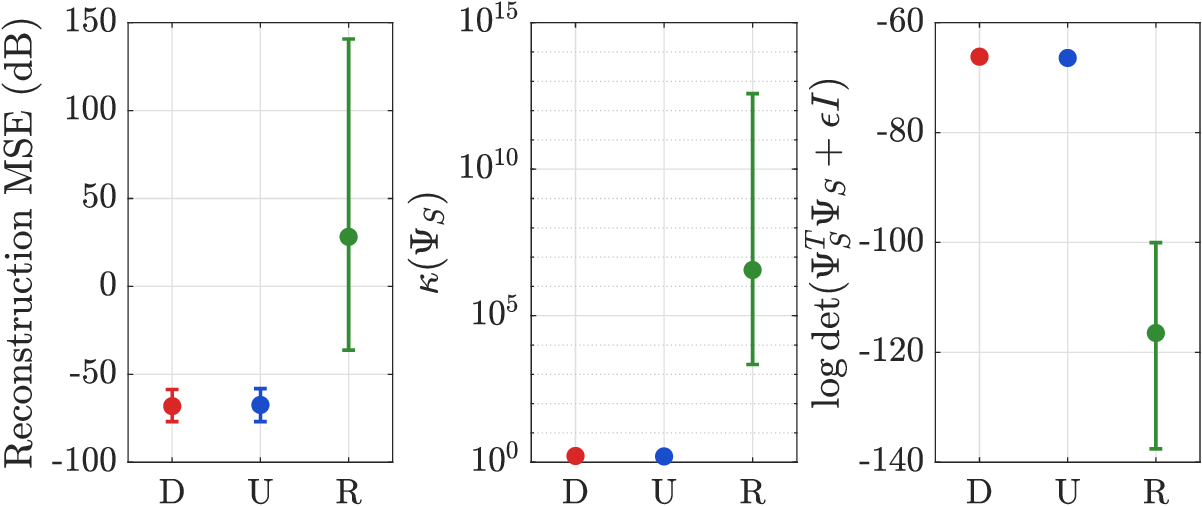}
    \caption{Regular-grid sampling diagnostics at $N_s=K=24$, including reconstruction MSE, condition number, and regularized log-determinant. D, U, and R denote D-optimal, uniform, and random sampling, respectively.}
    \label{fig:regular_design_metrics}
\end{figure}

Fig.~\ref{fig:regular_design_metrics} shows that, on the complete regular grid, D-optimal and uniform sampling are nearly indistinguishable in reconstruction accuracy. Their median MSE values are $-68.07$ and $-67.41$ dB, respectively. Their sampled basis matrices are also similarly well conditioned, with median $\kappa(\bm{\Psi}_{\mathcal S})$ values of 1.64 and 1.58, and their regularized log-determinants are close, at $-66.19$ and $-66.42$, respectively. In contrast, random sampling has a median MSE of 28.20 dB, a median condition number of $3.63\times10^{6}$, and a median log-determinant of $-116.48$. These results confirm that uniform sampling is a strong baseline in the regular-access case because it approximately preserves the near-orthogonality of the sampled low-frequency DCT columns.

\begin{figure}[t]
    \centering
    \includegraphics[width=\linewidth]{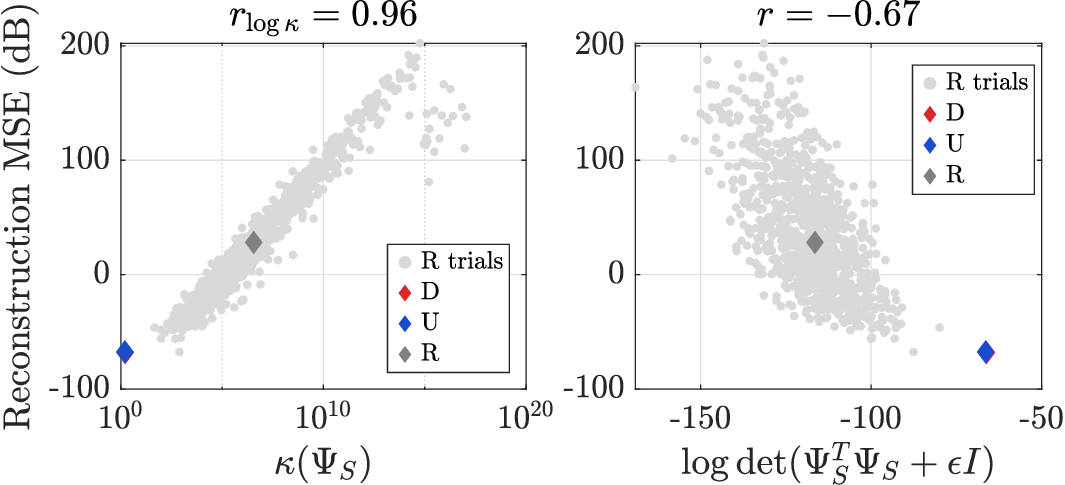}
    \caption{Relationship between reconstruction MSE and sampling-matrix metrics on the complete regular grid. The coefficient $r$ denotes the Pearson correlation coefficient computed over random-sampling trials.}
    \label{fig:regular_metric_relationship}
\end{figure}

Fig.~\ref{fig:regular_metric_relationship} further illustrates how sampling-matrix metrics relate to reconstruction error. Across random-sampling trials, the correlation between MSE in dB and $\log_{10}\kappa(\bm{\Psi}_{\mathcal S})$ is 0.964, whereas the correlation between MSE in dB and the regularized log-determinant is $-0.666$. Thus, ill-conditioned sampling matrices tend to amplify noise and model mismatch, while larger information volume tends to improve reconstruction stability. This relationship is statistical rather than deterministic, because the realized MSE also depends on the random antenna perturbation, the failed-element realization, the finite-$K$ model mismatch, and measurement noise. These observations suggest that optimized D-optimal sampling is most useful when physical constraints destroy the favorable uniform sampling geometry, which is examined next.

\subsection{Sampling Design with Orbit-Generated Opportunities}

\begin{table}[t]
\centering
\caption{Key parameters of the orbit-generated opportunity experiment.}
\label{tab:orbit_parameters}
\footnotesize
\setlength{\tabcolsep}{3pt}
\begin{tabular}{p{0.32\linewidth}p{0.60\linewidth}}
\toprule
Item & Setting \\
\midrule
Main orbit & 520-km circular LEO, $i=53.0^\circ$, RAAN $=0^\circ$, initial phase $=0^\circ$ \\
Calibration orbit & 500-km circular LEO, $i=53.5^\circ$, RAAN $=0.5^\circ$, initial phase $=2.5^\circ$ \\
Propagation & Two-body Keplerian model, 48~h campaign, 60-s sampling interval \\
Antenna frame & LVLH/RSW attitude, $x_{\rm A}=S$, $y_{\rm A}=W$, $z_{\rm A}=R$, $\theta=\operatorname{atan2}(y_{\rm A},x_{\rm A})$ \\
Screening & Earth clearance $\ge100$~km, range 50--1200~km, cone angle $\le10.5^\circ$, link margin $\ge3$~dB \\
Link budget & 26~GHz, 1~W transmit power, 35-dBi transmit/receive gains, 3-dB loss, 20-MHz bandwidth, 5-dB noise figure, 10-dB required SNR \\
\bottomrule
\end{tabular}
\end{table}

We finally examine the role of sampling design when the measurement opportunities are constrained by orbital geometry. In contrast to the complete regular angular grid considered in the previous subsection, the candidate directions here are generated from a simplified two-satellite orbital scenario and are generally nonuniform and incomplete. 

{The candidate set is generated using MATLAB's two-body Keplerian propagator with the parameters in Table~\ref{tab:orbit_parameters}. The calibration satellite is placed below the main satellite, while the differential altitude, inclination, RAAN, and phase make the relative line of sight sweep through the antenna-frame angular sector. Link margin is used for screening only; after compensating the known link terms, the retained power measurements follow the normalized noise model in \eqref{eq:normalized_link_budget_measurement}.}

\begin{figure}[t]
    \centering
    \includegraphics[width=0.98\linewidth]{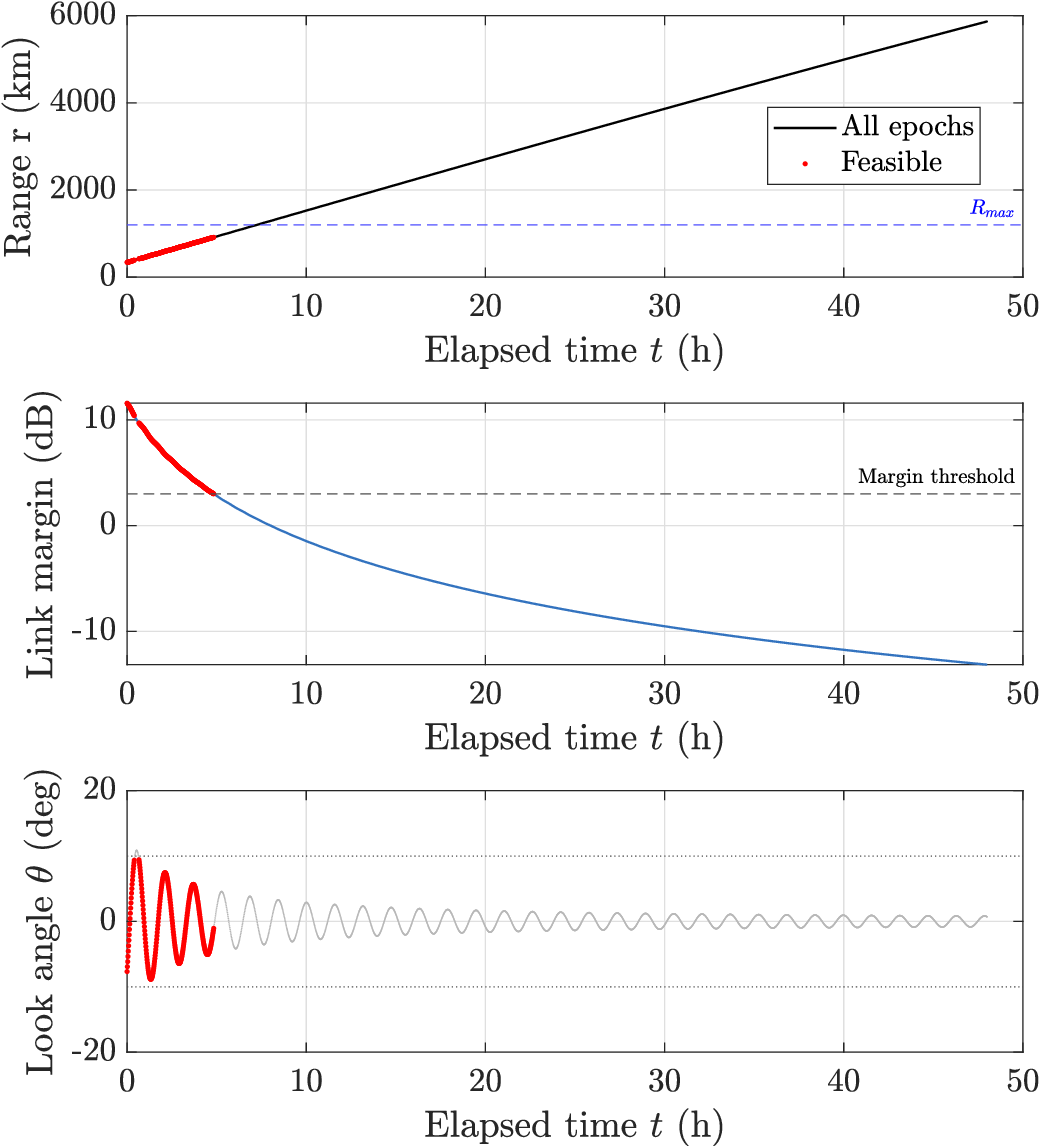}
    \caption{Orbit-generated opportunity screening over the 48-h calibration window. From top to bottom: inter-satellite range, link margin, and antenna-frame look angle. Red markers denote the 276 feasible epochs that pass all screening constraints.}
    \label{fig:orbit_geometry}
\end{figure}

\begin{figure}[t]
    \centering
    \includegraphics[width=0.98\linewidth]{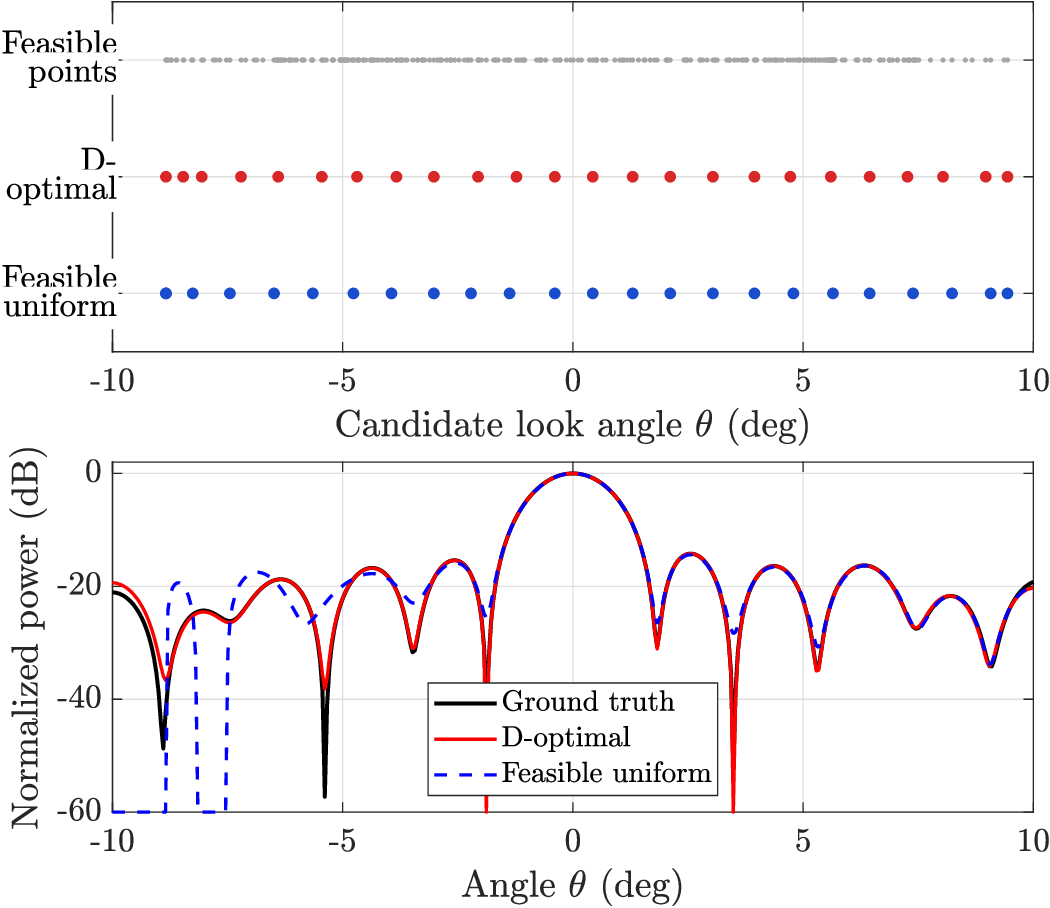}
    \caption{{Representative orbit-generated trial. Top: all feasible candidate angles and the angles selected by D-optimal and feasible uniform sampling. Bottom: the true and reconstructed normalized power patterns. }}
    \label{fig:orbit_selection}
\end{figure}

\begin{figure}[t]
    \centering
    \includegraphics[width=0.98\linewidth]{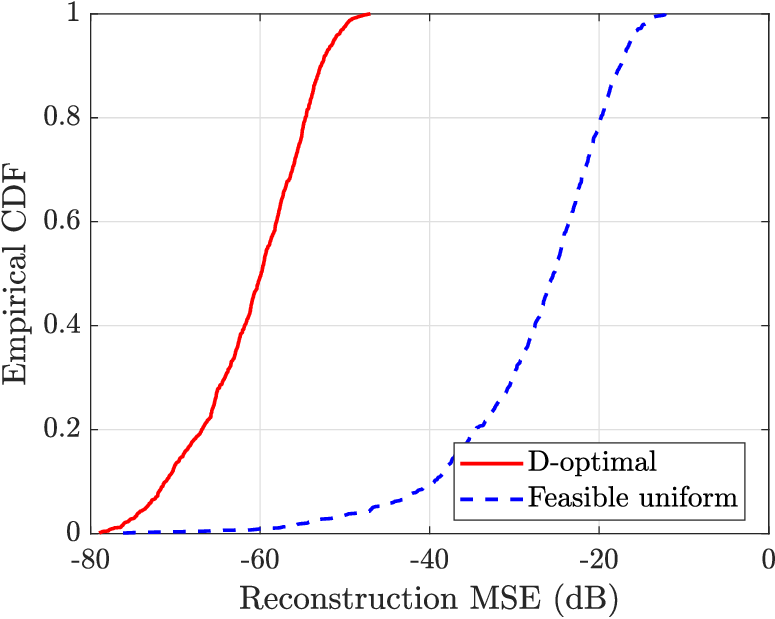}
    \caption{Empirical CDFs of reconstruction MSE for D-optimal and feasible uniform selection from the same 276 orbit-generated candidates, with $N_s=K=24$.}
    \label{fig:orbit_ecdf}
\end{figure}

Fig.~\ref{fig:orbit_geometry} on the next page illustrates the construction of the orbit-generated candidate set. The two-body propagation produces 2881 epochs over the 48-h calibration window. After applying Earth-visibility, range, link-margin, antenna-field-of-view, and pattern-sector constraints, 276 feasible measurement opportunities remain. These feasible epochs form the candidate set $\Omega_c$, from which $N_s=24$ measurements are selected. The feasible look angles span $[-8.84^\circ, 9.44^\circ]$, and the corresponding inter-satellite ranges are from 340.3 to 912.8 km. D-optimal and feasible uniform selection both choose 24 distinct epochs from this set, where feasible uniform greedily selects the remaining candidate nearest to each of 24 uniformly spaced targets over $[-10^\circ, 10^\circ]$.

{Figs.~\ref{fig:orbit_selection} and \ref{fig:orbit_ecdf} on the next page show that D-optimal selection leads to more stable reconstruction from the orbit-generated candidate set. In the representative trial, the MSE values are $-64.02$ and $-29.81$~dB for D-optimal and feasible uniform selection, respectively. Over 1000 trials, their median MSE values are $-59.87$ and $-25.39$~dB. The trial-wise median D-optimal gain is $34.21$~dB, and D-optimal selection achieves lower MSE in $98.6\%$ of trials. These results indicate that, for the irregular orbit-generated opportunity set, selecting samples according to the retained-subspace information criterion leads to more reliable reconstruction than simply approximating uniform angular targets from the feasible candidates.}

\section{Conclusion}
\label{sec:conclusion}
This paper develops a cooperative on-orbit framework for reconstructing satellite transmit antenna patterns from sparse directional measurements. We first convert the received calibration power into normalized pattern samples through orbit-attitude geometry and link-budget compensation. We then represent the antenna power pattern by a truncated DCT basis, thereby transforming the high-dimensional pattern recovery problem into a low-dimensional coefficient-estimation problem. A closed-form maximum-likelihood estimator is derived, and its error characterization reveals the effects of DCT truncation, measurement noise, and sampled-basis conditioning on reconstruction accuracy. Furthermore, we study the sampling design problem under different accessibility conditions. For regularly accessible angular sectors, uniform sampling provides an information-balanced baseline for the retained DCT modes; for constrained feasible opportunities, D-optimal sampling selects informative directions from the available candidates. Simulation results verify the sample efficiency of the low-dimensional DCT model and demonstrate the benefit of D-optimal sampling when orbit-generated opportunities are irregular or incomplete. These results show that reliable sparse on-orbit pattern reconstruction requires both a compact low-dimensional representation and a sampling strategy matched to the available calibration opportunities.

\bibliographystyle{IEEEtran}
\bibliography{IEEEabrv,references}

\end{document}